\documentclass[conference]{IEEEtran}
\usepackage{times}
\usepackage[usenames]{color}
\usepackage{graphicx}
\graphicspath{{./fig/}}
\usepackage[ruled,linesnumbered,vlined]{algorithm2e}
\usepackage{booktabs}
\usepackage{tabularx}
\usepackage{epsfig}
\usepackage{amssymb}
\usepackage{amsmath}
\usepackage{amsfonts}
\usepackage{multirow}

\newtheorem{theorem}{Theorem}[section]
\newtheorem{lemma}[theorem]{Lemma}

\ifCLASSINFOpdf
\else
\fi
\begin{document}
%
% paper title
% can use linebreaks \\ within to get better formatting as desired
\title{Compression and Sieve: Reducing Communication in Parallel Breadth
  First Search on Distributed Memory Systems}

% author names and affiliations
% use a multiple column layout for up to three different
% affiliations
%% \author{}
\author{\IEEEauthorblockN{Huiwei Lv\IEEEauthorrefmark{1}\IEEEauthorrefmark{2},
    Guangming Tan\IEEEauthorrefmark{1}, Mingyu Chen\IEEEauthorrefmark{1},
    Ninghui Sun\IEEEauthorrefmark{1}}
  \IEEEauthorblockA{\IEEEauthorrefmark{1}State Key Laboratory of Computer Architecture,\\
    Institute of Computing Technology, Chinese Academy of Sciences\\
    \IEEEauthorrefmark{2}Graduate University of Chinese Academy of Sciences\\
    Beijing, China 100190\\
    Email: lvhuiwei@ncic.ac.cn, tgm@ict.ac.cn, cmy@ict.ac.cn, snh@ncic.ac.cn} }

\maketitle

\begin{abstract}
%\boldmath
%% Breadth-first search (BFS) is a fundamental algorithm that forms the basis
%% of a great many graph algorithms.
  For parallel breadth first search (BFS) algorithm on large-scale distributed
  memory systems, communication often costs significantly more than arithmetic
  and limits the scalability of the algorithm. In this paper we sufficiently
  reduce the communication cost in distributed BFS by compressing and sieving
  the messages. First, we leverage a bitmap compression algorithm to reduce
  the size of messages before communication.  Second, we propose a novel
  distributed directory algorithm, cross directory, to sieve the redundant
  data in messages.  Experiments on a 6,144-core SMP cluster show our
  algorithm outperforms the baseline implementation in Graph500 by 2.2 times,
  reduces its communication time by 79.0\%, and achieves a performance rate of
  12.1 GTEPS (billion edge visits per second).
  %% and outperforms the combinational BLAS library at scale.
\end{abstract}
% IEEEtran.cls defaults to using nonbold math in the Abstract.
% This preserves the distinction between vectors and scalars. However,
% if the conference you are submitting to favors bold math in the abstract,
% then you can use LaTeX's standard command \boldmath at the very start
% of the abstract to achieve this. Many IEEE journals/conferences frown on
% math in the abstract anyway.

% no keywords

% For peer review papers, you can put extra information on the cover
% page as needed:
% \ifCLASSOPTIONpeerreview
% \begin{center} \bfseries EDICS Category: 3-BBND \end{center}
% \fi
%
% For peerreview papers, this IEEEtran command inserts a page break and
% creates the second title. It will be ignored for other modes.
\IEEEpeerreviewmaketitle

\section{Introduction}
\label{sec:intro}
Recently, graph has been extensively used to abstract complex systems and
interactions in emerging ``big data'' applications, such as social network
analysis, World Wide Web, biological systems and data mining. With the
increasing growth in these areas, petabyte-sized graph datasets are produced
for knowledge discovery~\cite{Bader:2007,Lumsdaine:2007}, which could only be
solved by distributed machines; benchmarks, algorithms and runtime systems for
distributed graph have gained much popularity in both academia and
industry~\cite{graph500,Yoo:2005,Buluc:2011,Malewicz:2009}. One of the most
widely used graph-searching algorithms is breadth-first search (BFS), which
serves as a building block for a great many graph algorithms such as minimum
spanning tree, betweenness centrality, and shortest
paths~\cite{Chazelle:2000,tan-cyclops64,bc-brandes,Cherkassky:1996}.

Implementing a distributed BFS with high performance, however, is a
challenging task because of its expensive communication
cost~\cite{Chan05cgmgraph,Lumsdaine:2007}. Generally, algorithms have two
kinds of costs: arithmetic and communication. For distributed algorithms,
communication often costs significantly more than arithmetic. For example, on
a 512-node cluster, the baseline BFS algorithm in Graph 500 spends about 70\%
time on communication during its traversal on a scale-free graph with 8
billion vertices (Figure~\ref{fig:intro-percent}).  Therefore the most
critical task in a distributed BFS algorithm is to minimize its communication.

\begin{figure}[t]
  \centering
  \includegraphics[width=0.4\textwidth]{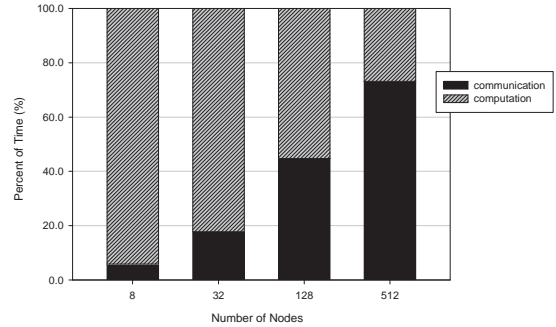}
  \caption{Time breakdown of a baseline distributed BFS in a weak scaling
    experiment that use fixed problem size per node (each node has about 16M
    vertices).}
  \label{fig:intro-percent}
\end{figure}
\begin{table}[t]
  \caption{Comparison of various approaches for reducing communication cost in
    distributed BFS.}
  \label{table:comparison}
  %%\newcolumntype {Z} {>{\centering\arraybackslash}X}
  \begin{center}
     \begin{tabularx}{0.45\textwidth}{ll}
       \toprule%
       Approach & Category \\
       \midrule
       Two-dimensional partitioning~\cite{Yoo:2005,Buluc:2011} & algorithm \\
       Bitmap \& sparse vector~\cite{graph500,Buluc:2011} & data structure \\
       PGAS with communication coalescing~\cite{Cong:2010} & runtime \\
       \midrule
       \textit{This Work}: compression \& sieve & data structure \\
       \bottomrule
    \end{tabularx}
  \end{center}
\end{table}

%% Related Works: spend some time comparing and contrasting the paper to other
%% literature, and demonstrate why the paper’s results and techniques are
%% new, interesting, and/or surprising given this context.
Several different approaches are proposed to optimize communication in
distributed BFS (Table~\ref{table:comparison}): using two-dimensional
partitioning of the graph to reduce communication
overhead~\cite{Yoo:2005,Buluc:2011}, using bitmap or sparse vector to reduce
the size of messages~\cite{graph500,Buluc:2011}, or applying communication
coalescing in PGAS implementation to minimize message
overhead~\cite{Cong:2010}. These approaches attack the problem from different
angles: algorithm, data structure and runtime. In this paper, we will focus on
reducing the size of communication messages (the optimization of data
structures).  The main techniques we use are \emph{compression} and
\emph{sieve}. Overall, we make the following contributions:
\begin{itemize}
\item By compressing the messages, we reduce the communication time by $52.4\%$
  and improved its overall performance by $1.7\times$ compared to the baseline BFS
  algorithm. 
%%(Section~\ref{sec:spmv-bfs})
\item By sieving the messages with a novel distributed directory before
  compression. We further reduce the communication by $55.9\%$ and improved the
  performance by another $1.3 \times$, achieving a total $79.0\%$ reduction in
  communication and $2.2 \times$ performance improvement over the baseline
  implementation. 
%%(Section~\ref{sec:dir-bfs})
\item We implement and analyse several compression methods for bitmap
  compression. Our experiment shows the space-time tradeoff of different
  compression methods. 
%%(Section~\ref{sec:exp-res})
\end{itemize}

In the next section we will introduce the problem with an
example. Section~\ref{sec:baseline-bfs} will describe the baseline BFS
algorithm. Section~\ref{sec:spmv-bfs} and Section~\ref{sec:dir-bfs} will
describe our BFS algorithms with compression and sieve. The analysis and
experiment results are presented in Section~\ref{sec:ana} and
Section~\ref{sec:exp}, followed by related works and concluding remarks in
Section~\ref{sec:related} and Section~\ref{sec:cls}.

\section{Motivation}
\label{sec:pre}

\begin{figure}[t]
  \centering
  \includegraphics[width=0.48\textwidth]{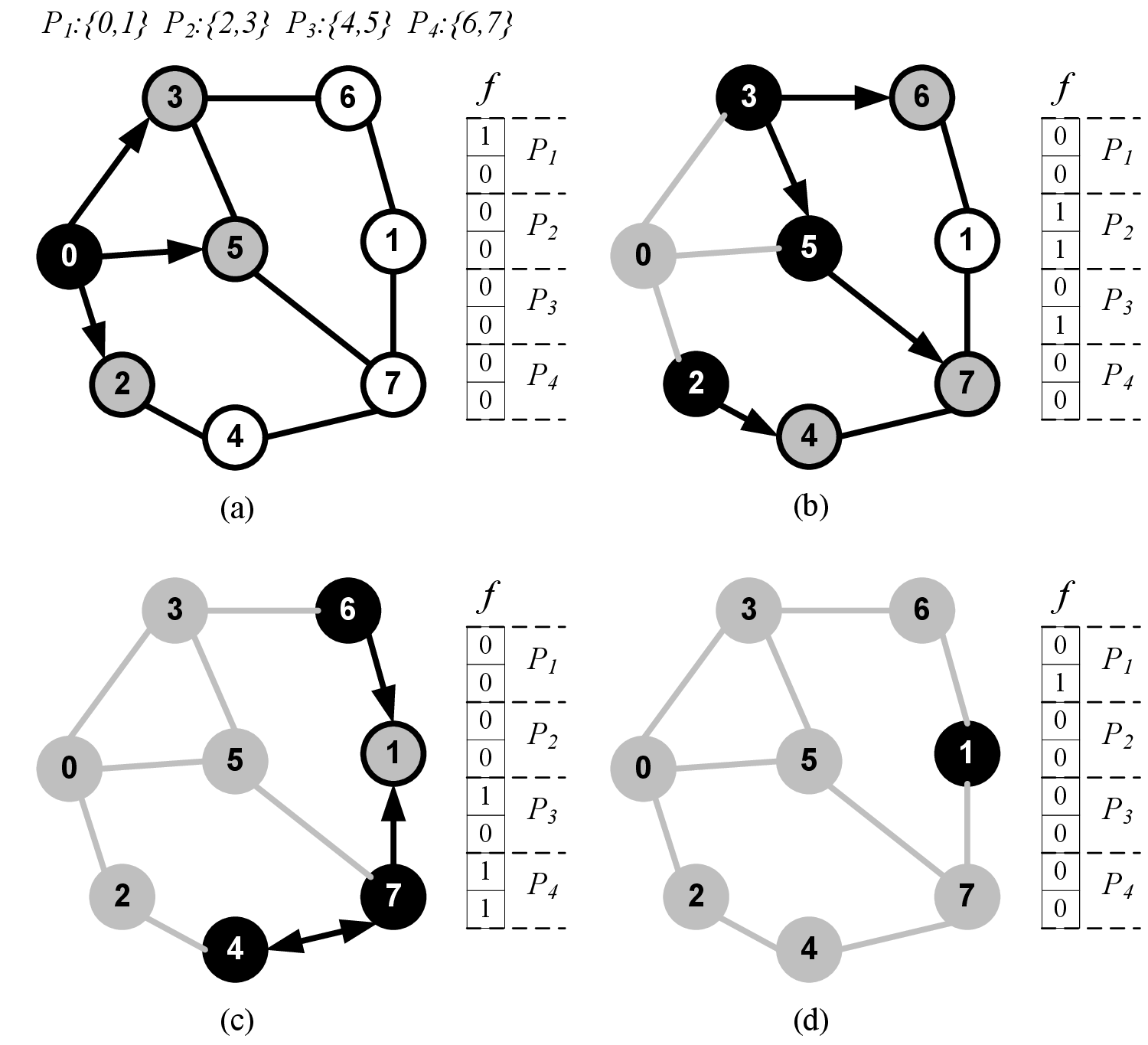}
  \caption{The operation of BFS on an undirected graph. The frontier $f$ is
    represented as a vector.}
  \label{fig:bfs-example-all}
\end{figure}
We start with an example illustrating the breadth-first search (BFS)
algorithm. Given a graph $G = (V, E)$ and a distinguished source vertex $s$,
breadth-first search systematically explores the edges of $G$ to ``discover''
every vertex that is reachable from $s$. In Figure~\ref{fig:bfs-example-all},
the source vertex $0$ is painted black when the algorithm begins. Then it
explores its adjacent vertices: $3$, $5$ and $2$, and paints them black. The
exploration goes on until all vertices are visited. Vertices discovered the
first time is painted black; discovered vertices are painted solid grey;
vertices to be discovered are painted grey with black edge. The frontier $f$
of the graph is the set of the vertices which are discovered the first time.

For distributed BFS, the vertices as well as the frontier are divided among
processors: $P_1:\{0,1\}$, $P_2:\{2,3\}$, $P_3:\{4,5\}$, $P_4:\{6,7\}$. And
the global information of the frontier can only be retrieved through
communication. For $P_1$ in this example, it only ``owns'' the information of
whether vertex $0$ and $1$ are visited. If it want to identify whether vertex
$2$ is visited, it needs to ask this information from $P_2$. The common way to
update the global $f$ is to use MPI collective communication like
\textsc{Allgather} at the end of each
level~\cite{Yoo:2005,Buluc:2011,graph500}.

The most critical task for distributed BFS is to reduce the size of the
frontier, which directly influence the size of the messages communicated.  To
reduce it, bitmap or sparse vector is commonly used to represent the
frontier. Bitmap use a vector of size $|V|$ to represent the frontier, each
bit of the vector representing a vertex: $1$ means it is included in the
frontier, $0$ means it is not. Sparse vector includes the frontier vertices
only, each is represented using 64 bits. For graphs of diameter $d$, bitmap is
generally better when $d < 64$. Table~\ref{table:sparsity} provides an example
of the size of the frontier represented as bitmap or sparse vector, for a
scale-free graph of 1.6 billion vertices. In this case, for $d=7$, the total
size of messages using bitmap is 1.4 GB, much less than the sparse vector's
12.4 GB.
\begin{table}[t]
  \caption{The size of the frontier, represented as bitmap or sparse vector,
    at each level of BFS of a scale-free graph with 1.6 billion vertices. For
    sparse vector, each vertex is represented as a 64-bit number.}
  %%\vspace{-0.8cm}
  \label{table:sparsity}
  \newcolumntype {Z} {>{\centering\arraybackslash}X}
  \begin{center}
     \begin{tabularx}{0.45\textwidth}{lZZZ}
       \toprule%
       Level & \#Vertices & bitmap & sparse vector\\
       \midrule
       1 & 2 &  196.9MB & 16B\\
       2 & 20842  & 196.9MB & 162.8KB\\
       3 & 235274348  & 196.9MB & 2.0GB\\
       4 & 1377666413  & 196.9MB & 10.2GB\\
       5 & 38582585  & 196.9MB & 294.4MB\\
       6 & 88639 & 196.9MB & 692.4KB\\
       7 & 211 & 196.9MB & 1.69KB\\
       \midrule
       Total & 1651633040 & 1.4GB & 12.4GB \\
       \bottomrule
    \end{tabularx}
  \end{center}
\end{table}
Despite the huge space saved by bitmap, there remains two problems:
\begin{itemize}
\item The problem of bitmap is that it need to contain \emph{all} the vertices
  to keep the position information of each vertex. For the above example, to
  represent 2 vertices at level 1, the size of the bitmap frontier is still
  196.9MB, where most of the elements are zero. Fortunately, these zeros can
  be condensed.  \emph{We leverage lossless compression to reduce the size of
    the bitmap}.
\item The other problem is the expensive broadcast cost of the
  \textsc{Allgather} collective communication, which broadcasts \emph{all}
  vertices to \emph{all} processors. In fact, each processor needs only a
  small fractions of the frontier.  For example, in
  Figure~\ref{fig:bfs-example-all} (b), $P_2$ does not need to send the
  information of vertex $2$ to $P_4$, because vertex $2$ does not has a direct
  edge connecting to the vertices of $P_4$.  \emph{We propose a distributed
    directory to sieve the bitmap vectors before compression, further reducing
    its message size}.
\end{itemize}

\section{Baseline BFS with Bitmap as Frontier}
\label{sec:baseline-bfs}

\subsection{BFS Described in Linear Algebra}
\label{sec:preliminary}

Let $A$ denote the adjacency matrix of the graph $G$, $f_{Lk}$ denote the
frontier at level $k$, and $\pi_k = \bigcup^{k}_{i=1}f_{Li}$ denote the
visited information of previous frontiers. The exploration of level $k$ in BFS
is algebraically equivalent to a sparse matrix vector multiplication (SpMV):
$f_{L(k+1)} \gets A^T \otimes f_{Lk} \odot \overline{\pi_k}$ (we will omit the
transpose and assume that the input is pre-transposed for the rest of this
section).
%% Each BFS iteration is computationally equivalent to a sparse matrix-sparse
%% vector multiplication (SpMSV).  
For example, traversing from level one (Figure~\ref{fig:bfs-example-all} (a))
to level two (Figure~\ref{fig:bfs-example-all} (b)) is equivalent to the
linear algebra below.
\[
A^T \otimes f_{L0} \odot \overline{\pi_0} =
 \begin{pmatrix}
   0 0 1 1 0 1 0 0\\
   0 0 0 0 0 0 1 1\\
   1 0 0 0 1 0 0 0\\
   1 0 0 0 0 1 1 0\\
   0 0 1 0 0 0 0 1\\
   1 0 0 1 0 0 0 1\\
   0 1 0 1 0 0 0 0\\
   0 1 0 0 1 1 0 0\\
 \end{pmatrix}
\otimes
 \begin{pmatrix}
   1\\   0\\   0\\   0\\   0\\   0\\   0\\   0\\
 \end{pmatrix}
\odot
 \begin{pmatrix}
   0\\   1\\   1\\   1\\   1\\   1\\   1\\   1\\
 \end{pmatrix}
=
 \begin{pmatrix}
   0\\   0\\   1\\   1\\   0\\   1\\   0\\   0\\
 \end{pmatrix}
= f_{L1}
\]
The syntax $\otimes$ denotes the
matrix-vector multiplication operation, $\odot$ denotes element-wise
multiplication, $(a_1, a_2, \cdots, a_n)^T \odot (b_1, b_2, \cdots, b_n)^T =
(a_1b_1, a_2b_2, \cdots, a_nb_n)^T$, and overline represents the complement
operation. In other words, $\overline{v_i} = 0$ for $v_i \neq 0$ and
$\overline{v_i} = 1$ for $v_i = 0$.

In Figure~\ref{fig:bfs-example-all}, BFS
starts from vertex $v_0$, thus $f_{L0}=\{v_0\},
f_{L1}=\{v_2,v_3,v_5\},f_{L2}=\{v_4,v_6,v_7\},f_{L3}=\{v_1\}$. If we use a
vector of size $n$ to represent the corresponding frontier $f_{Lk}$, for
example, $f_{L2}=\{0,0,1,1,0,1,0,0\}$. This algorithm
becomes deterministic with the use of (select, max)-semiring, because the
parent is always chose to be the vertex with the highest label. 

\subsection{Baseline BFS}
\label{sec:baseline-bfs-alg}

\begin{algorithm} [t] %[H] %[t]
\label{alg:spmv-bfs}
\SetKwInOut{Input}{Input}
\SetKwInOut{Output}{Output}
%% \Input{A: undirected graph represented by a boolean sparse adjacency matrix,
%% s: source vertex id}
%% \Output{$\pi$: dense vector, where $\pi[v]$ is the predecessor vertex on the
%% shortest path from $s$ to $v$, or $-1$ if $v$ is unreachable.}
\Input{s: source vertex id}
\caption{A baseline distributed BFS}
$f(s) \gets s$\;
\ForEach{processor $P_i$ in parallel}
{
  \While{$f \neq \emptyset$}
  {
    $t_{i} \gets A_{i} \otimes f $\;
    $t_{i} \gets t_{i} \odot \overline{\pi_{i}}$; $\pi_{i} \gets \pi_{i} + t_{i}$\;
    $f_{i} \gets t_{i}$\;
    $f \gets \textsc{Allgatherv}(f_i,P_i)$\;
  }
}
\end{algorithm}
Algorithm \ref{alg:spmv-bfs} describes the baseline BFS. Each loop block
(starting in line 3) performs a single level traversal. $f$ represents the
current frontier, which is initialized as an empty bitmap; $t$ is an bitmap
that holds the temporary parent information for that iteration only; $\pi$ is
the visited information of previous frontiers. The computational step (line
4,5,6) can be efficiently parallelized with multithreading. For SpMV operation
in line 4, the matrix data is naturally splitted into pieces for
multithreading. At the end of each loop, \textsc{Allgather} updates $f$ with
MPI collective communication.

\section{BFS with Compression}
\label{sec:spmv-bfs}

For large graphs, the communication time of distributed BFS algorithms can
take as much as seventy percent of the total execution time. To reduce it, we
need to reduce the size of the messages. One simple way is to use lossless
compression, trading computation for bandwidth. 
%% The benefit of this approach will depend on the compression ratio and the
%% compression speed.

%% \subsection{BFS with compression}
%% \label{sec:bfs-compression}

\begin{algorithm} [t] %[H] %[t]
\label{alg:bfs-compress}
\caption{Distributed BFS with compression.}
$f(s) \gets s$\;
\ForEach{processor $P_i$ in parallel}
{
  \While{$f \neq \emptyset$}
  {
    $t_{i} \gets A_{i} \otimes f $\;
    $t_{i} \gets t_{i} \odot \overline{\pi_{i}}$\;
    $\pi_{i} \gets \pi_{i} + t_{i}$; $f_{i} \gets t_{i}$\;
    $f_{i}' \gets Compress(f_{i})$\;
    $f' \gets \textsc{Allgatherv}(f_i',P_i)$\;
    $f \gets Uncompress(f')$\;
  }
}
\end{algorithm}

Algorithm~\ref{alg:bfs-compress} describe the distributed BFS with
compression. The difference between Algorithm~\ref{alg:bfs-compress} and
Algorithm~\ref{alg:spmv-bfs} are line 7 and 9. At line 7 the frontier vector
$f$ is first compressed into $f'$ before communication. At line 9 $f'$ is
uncompressed back to $f$ after communication.

\begin{table}[t]
  \caption{A WAH compressed bitmap.}
  \label{table:wah-example}
  \begin{center}
    \begin{tabular}{ll}
%%       \hline
%%       128 bits & \multicolumn{4}{c}{1*1, 20*0, 3*1, 79*0, 25*1} \\
%%       \hline
%%       31-bit groups & 1,20*0,3*1,7*0 & 62*0 & 10*0, 21*1 & 4*1 \\
%%       literal (hex) & 40000380 & 00000000 00000000 & 001FFFFF & 0000000F \\
%%       WAH (hex) & 40000380 & 80000002 & 001FFFFF & 0000000F \\
%%       \hline
      \hline
      16 bits & 1000000000000000 \\
      3-bit groups & 100 000 000 000 000 0\\
      WAH & 0100 1100 0000\\
      \hline
    \end{tabular}
  \end{center}
\end{table}
We use word-aligned hybrid (WAH)~\cite{Wu:WAH} for \textit{Compress} and
\textit{Uncompress} function, as WAH is fast and well suited for bitmap
compression. Table~\ref{table:wah-example} shows the WAH compressed
representation of 16 bits. In WAH, there are three types of words: literal
words, fill words and active words. The most significant bit of a word is used
to distinguish between a literal word (0) and a fill word (1). And a active
word stores the last few bits. We assume that each computer word contains 4
bits and all fill bits are 0 in this example. Under this assumption, each
literal word stores 3 bits from the bitmap, and each fill word represents a
multiple of 3 bits. The second line in Table~\ref{table:wah-example} shows the
bitmap as 3-bit groups. The last line shows the WAH words. The first two words
are regular words, the first is a literal word, and the second a fill
word. The fill word $1100$ indicates a 0-fill of 4 words long (containing 12
consecutive 0 bits). Note that the fill word stores the fill length as 4
rather than 12. The third word is the active word; it stores the last few bits
that could not be stored in a regular word. For sparse bitmaps, where most of
the bits are 0, a WAH compressed bitmap would consist of pairs of a fill word
and a literal word~\cite{Wu:WAH}.

Other lossless compression methods include run-length encoding, huffman
coding, LZ77~\cite{lz77}, or more dedicated bitmap compression method such as
byte-aligned bitmap compression (BBC)~\cite{Antoshenkov:1995:BBC} and position
list word aligned hybrid (PLWAH)~\cite{Deliege:2010:PLWAH}. There is a
space-time tradeoff among these compression schemes. Comparing to WAH, LZ77 is
slower but has a better compression ratio. The benefit of compression will
depend on many factors such as compression ratio, sparsity of the messages,
compression speed and network bandwidth. The best compression scheme can not
be determined beforehand, so we use experiment to analyse these
tradeoffs. Details will be presented in Section~\ref{sec:exp}.
%% 引WAH论文说明，压缩时间比ZLB快，压缩比比ZLB差。有一个tradeoff，那个更优通过实
%% 验来验证。

\begin{figure}[t]
  \centering
  \includegraphics[width=0.5\textwidth]{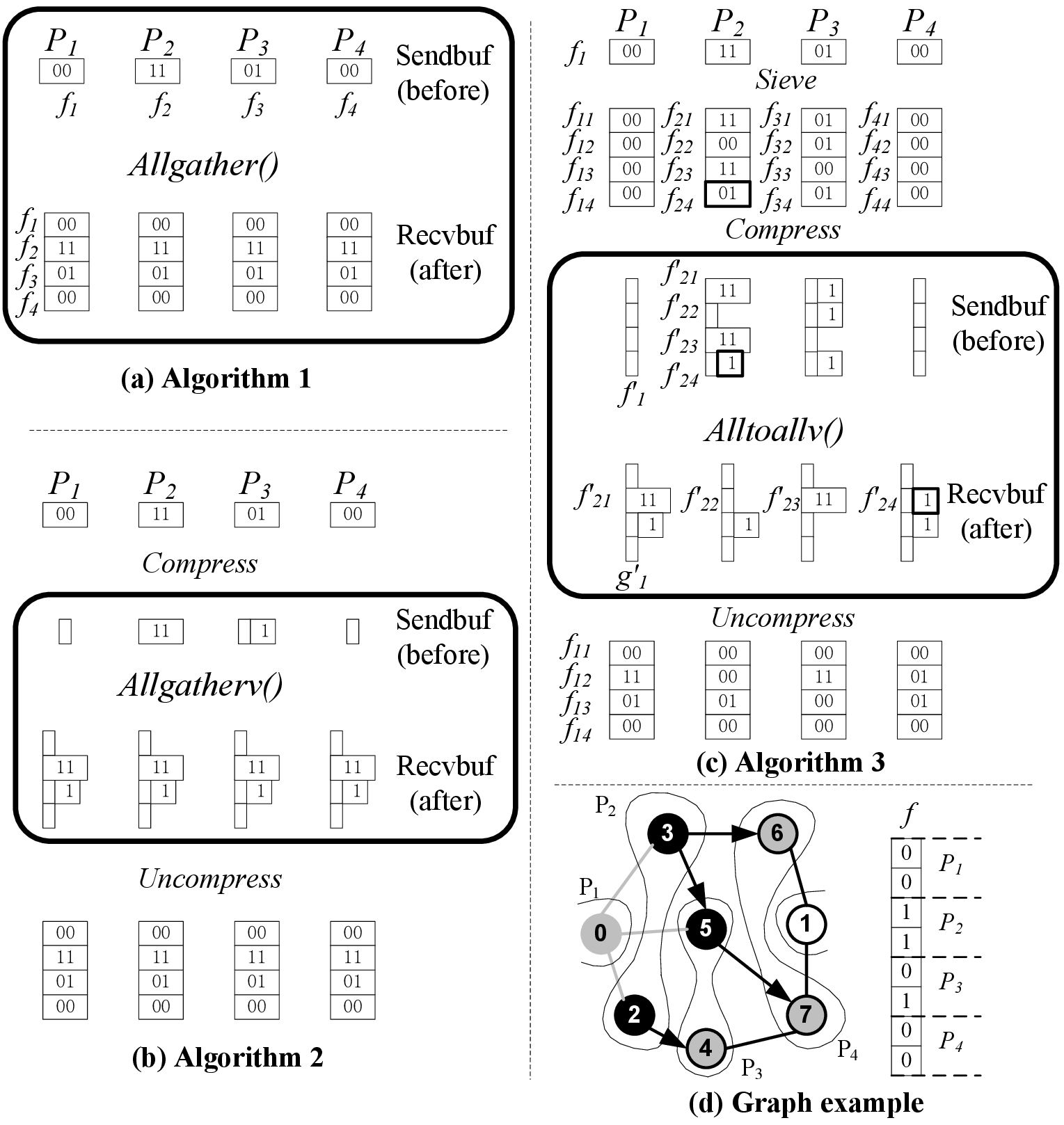}
  \caption{Three different ways of communication.}
  \label{fig:compress}
\end{figure}

\section{BFS with Compression and Sieve}
\label{sec:dir-bfs}

The message size in the communication is reduced after compression. But there
is still room for improvement. To achieve a better compression ratio, we can
use a directory to sieve the bitmap, making it even sparser for compression.

In this section we propose a distributed directory, \textit{cross directory},
as a sieve to reduce the number of messages sent to each processor. We will
first introduce the data structure of cross directory in
subsection~\ref{sec:cross-dir}, then describe our BFS algorithm with
compression and sieve in subsection~\ref{sec:dir-bfs-alg}.

\subsection{Cross Directory}
\label{sec:cross-dir}

The problem of collective communication like \textsc{Allgatherv} is that it
sends all frontier vertices to all the processors --- just like snoopy cache
coherence algorithms, all updates are visible to all processors --- regardless
whether a vertex is meaningful to each processor. Take a look at
Figure~\ref{fig:compress} (a), after the \textsc{Allgather} communication,
each processor actually get all the frontier vectors. In fact, \emph{each
  processor needs only a small fraction of the frontier, and this fraction can
  be determined before communication}. For example, in
Figure~\ref{fig:compress} (d), $P_4$ only needs $v_3$ from $P_2$. This means
$P_2$ does not need to send the information of $v_2$ to $P_4$, because $v_2$
does not has a direct edge connecting to the vertices of $P_4$.

To explain this in algebra, we first partition the matrix $A$ into $p$
block-rows. Then partition each block $A_i$ into $p$ sub-blocks.
\begin{equation}
  \label{eq:abc}
  A \otimes f =
 \begin{pmatrix}   A_1 \\   A_2 \\   \vdots \\   A_p \end{pmatrix}
 \begin{pmatrix}   f_1 \\   f_2 \\   \vdots \\   f_p \end{pmatrix}
\end{equation}
\begin{equation}
  \label{eq:sub-block}
  A_i = \begin{pmatrix} A_{i,1} & A_{i,2} & \cdots & A_{i,p} \end{pmatrix}
\end{equation}
To calculate $f_4 = \sum_{i=1}^{4} A_{4,i} \otimes f_i$, 
\begin{equation}
  \label{eq:a4}
  A_{4,2} \otimes f_2 = \begin{pmatrix} 01 \\ 00 \end{pmatrix} 
  \begin{pmatrix} x_1 \\ x_2 \end{pmatrix} = \begin{pmatrix} y_1 \\ y_2 \end{pmatrix},
\end{equation}
because $a_{0,0}$ and $a_{1,0}$ of $A_{4,2}$ are always zero (denote
$A_{i,j}=[a_{i,j}]_{m \times n}$), $y_1$ will always be zero. So $P_2$ does
not need to send $x_1$ to $P_4$. We define a data structure to record this
information and use it to sieve communication messages.

We formally define \textit{directory vector} as follows: for each item $v_k$
in vector $V_{i,j}$, $v_k$ is set to one if column $k$ in $A_{i,j}$ contains
at least one non-zero.
\begin{align}
  \label{eq:dir-vec}
V_{i,j} &= (v_1,v_2,\cdots,v_n) \\
& where \ v_k =
  \begin{cases}
1, & \exists a_{i,k}=1, i \in [1,m], k \in [1,n] \\
0, & otherwise
  \end{cases}\nonumber
\end{align}
For the above example, $V_{4,2}=(0,1)$ is sent to $P_2$ from $P_4$ during
initialization. When traversing begins, $f_2$ is sieved into $f_{2,4}=f_2 \odot
V_{4,2}=(1,1)^T \odot (0,1)^T = (0,1)^T$, so we only send one vertex (in
compressed bitmap format) back instead of two. This ``sieve effect'' is where
communication is reduced.

\begin{figure}[t]
\centering
  \includegraphics[width=0.3\textwidth]{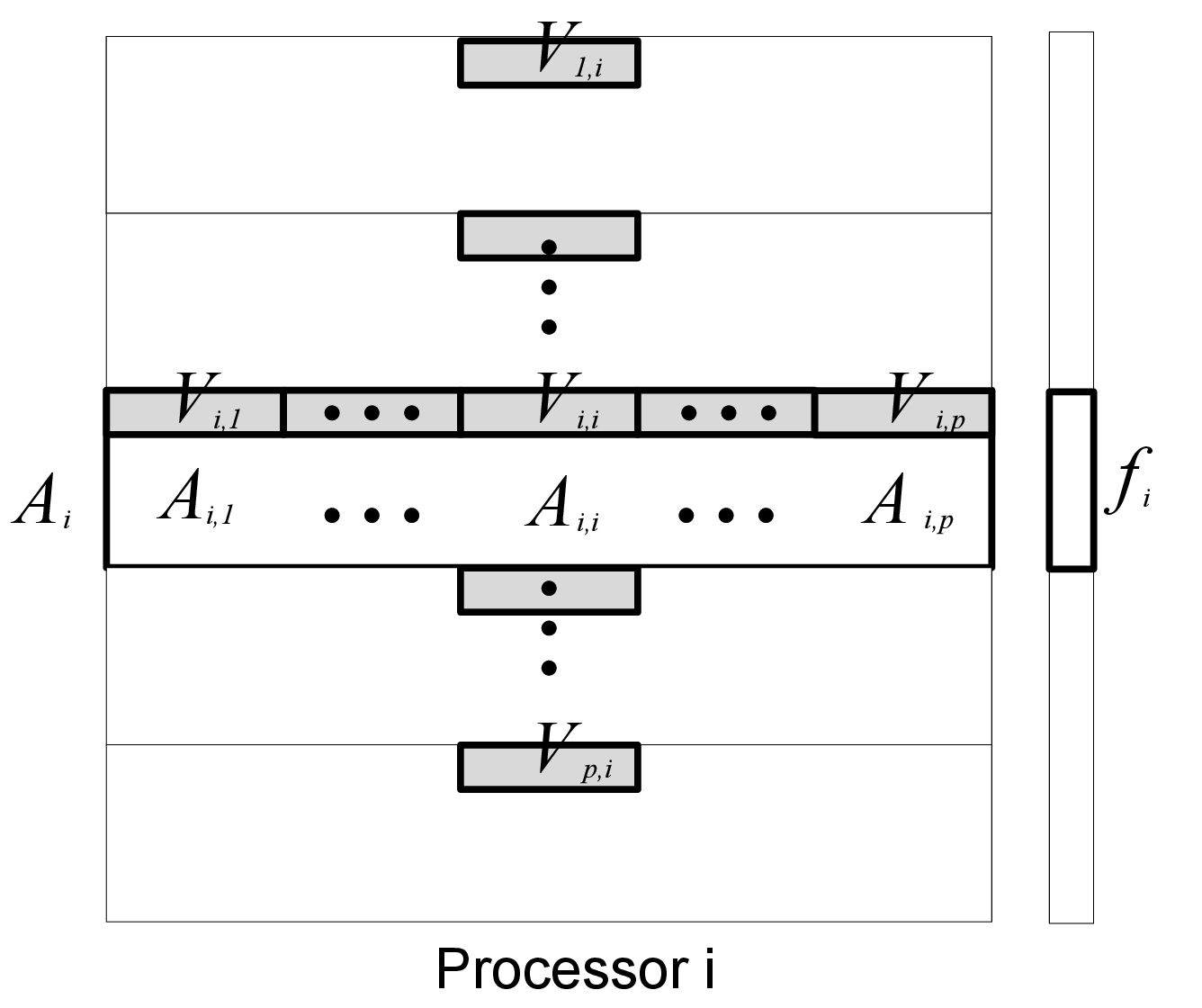}
  \caption{The \textit{cross directory} data structure for processor $i$.}
  \label{fig:dir-ds}
\end{figure}

And the \textit{cross directory} of processor $P_i$ is defined as:
\begin{equation}
  \label{eq:cross-dir}
  \mathbb{C}_i = \{ V_{x,i} \ or\  V_{i,x} \mid x=1,2,\cdots,p \}
\end{equation}
Besides a row of directory vectors $V_i=\{V_{i,y} \mid y=1,2,\cdots,p\}$,
$P_i$ own a copy of the directory vectors $\{V_{x,i} \mid x=1,2,\cdots,p\}$ in
column $i$. The directory in the column direction is established during
initialization and used to provide a local lookup for sieving (See
Figure~\ref{fig:dir-ds}).

\begin{figure}[t]
  \centering
  \includegraphics[width=0.45\textwidth]{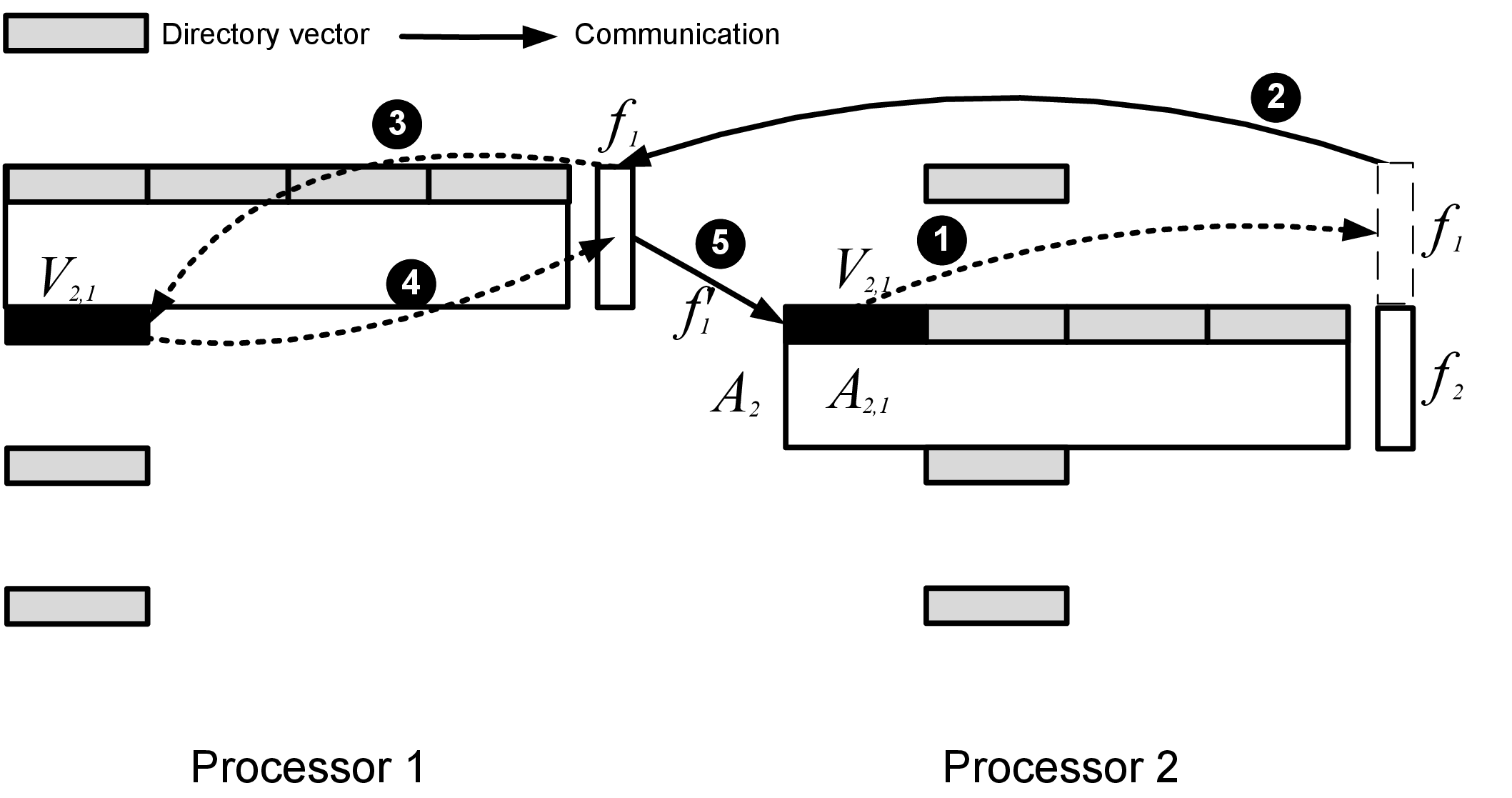}
  \caption{Communications in Directory-based BFS algorithm. Example of
    multiply $A_{2,1}$ with $f_1$ in five steps.}
  \label{fig:cross-directory}
\end{figure}

Figure~\ref{fig:cross-directory} illustrates an example of
communication with cross directory. The matrix is row-block partitioned among
four processors.
%%  Each processor owns an directory vector to indicate the
%% non-zeros in its matrix block and an additional copy of the directory vector
%% along its column. Following we give an example of
$A_{2,1} \otimes f_1$ is done in five steps, $A_{2,1}$ need to get $f_1$ (step
1), $P_2$ then send a request message to $P_1$ (step 2), $P_1$ check its local
copy of $V_{2,1}$ (step 3) and sieve $f_1$ with the non-zero positions (step
4), then $P_1$ send back a sieved $f_1'$ (step 5). The sieved vector is
very sparse and can be represented as sparse vector, reducing the
communication cost.

\subsection{Sieve with Cross Directory}
\label{sec:dir-bfs-alg}

Algorithm~\ref{alg:bfs-dir-compress} is our directory-based algorithm with
compression and sieve: based on Algorithm~\ref{alg:spmv-bfs},
Algorithm~\ref{alg:bfs-dir-compress} first sieves the frontier bitmap with the
cross directory (line 9), making it sparser; then it compresses this sieved
bitmap (line 10) and send it with \textsc{Alltoallv} (line 11); after received
the compressed bitmap, the original vector could be restored with
uncompression (line 13).

\begin{algorithm} [t] %[H] %[t]
\label{alg:bfs-dir-compress}
\SetKwInOut{Input}{Input}
\SetKwInOut{Output}{Output}
\KwData{$f_{i}' = \{f_{i,1}',f_{i,2}',\cdots,f_{i,n-1}'\}$:send buffer;
$g_{i}' = \{g_{i,1}',g_{i,2}',\cdots,g_{i,n-1}'\}$:receive buffer;
$\mathbb{C}_i$:cross directory for $P_i$.}
\caption{Distributed BFS with sieving and compression.}
$f(s) \gets s$\;
initialize $\mathbb{C}_i$\;
\ForEach{processor $P_i$ in parallel}
{
  \While{$f \neq \emptyset$}
  {
    $t_{i} \gets \sum_{j=1}^{n} A_{i,j} \otimes f_{i,j} $\;
    $t_{i} \gets t_{i} \odot \overline{\pi_{i}}$\;
    $\pi_{i} \gets \pi_{i} + t_{i}$; $f_{i} \gets t_{i}$\;
    \ForEach{$j \in [0,n)$ in parallel}
    {
      $f_{i,j} = f_i \odot V_{j,i}$;\tcc*[f]{sieving}\;
      $f_{i,j}' \gets Compress(f_{i,j})$\;
    }
    $g_i' \gets \textsc{Alltoallv}(f_i',P_i)$\;
    \ForEach{$j \in [0,n)$ in parallel}
    {
      $f_{i,j} \gets Uncompress(g_{i,j}')$\;
    }
  }
}
\end{algorithm}

This \textit{cross directory} is inspired by Pinar and Hendrickson's
distributed directory~\cite{Pinar:2001} and Baker et al.'s assumed partition
algorithm~\cite{Baker:2006}. In their work, the communication pattern is
dynamically determined and more general, while in our case, the communication
parties are static. So we store the directory on both side of the
communication, and update them synchronously on each side instead of send the
updated directory over the network each time.  Another difference lies in the
collective communication. In Baker et al.'s assumed partition algorithm,
point-to-point rendezvous communication is used, we find that could be
replaced with a more efficient \textsc{Alltoallv}. More generally, the cross
directory is applicable to matrix-vector multiplication when following
premises are true: 1) the partition of the matrix is static so that
communication parties are static; 2) the matrix remains unchanged and
multiplication takes many times so that cross directory could be reused and
its initialization cost could be omitted. For example, sum of the
multiplication of the same matrix with different vectors
($\sum_{i=1}^{n}Ax_i$).

\subsection{Proof of Correctness}
\label{sec:prof-correct}

In this subsection we prove the correctness of
Algorithm~\ref{alg:bfs-dir-compress} by proving its equivalence to
Algorithm~\ref{alg:spmv-bfs}.

\begin{lemma}\label{lemma-1}
  $A_{i,j} \otimes f_j = A_{i,j} \otimes f_j \odot V_{j,i} $.
\end{lemma}
\begin{proof}
  Let $X=(x_1,x_2,\cdots,x_n)^T=A_{i,j} \otimes f_j$,
  $Y=(y_1,y_2,\cdots,y_n)^T=A_{i,j} \otimes f_j \odot V_{j,i}
  =(x_1v_1,x_2v_2,\cdots,x_nv_n)^T$,$V_{j,i} = (v_1,v_2,\cdots,v_n)^T$, and
  $f_j=(z_1,z_2,\cdots,z_n)^T$. Denote $A_{i,j}=[a_{i,j}]_{m \times n}$, then
  $x_k = \sum_{l=1}^{n} a_{k,l}z_l$. According to the definition of directory
  vector, if $v_k = 0 \Rightarrow \forall a_{k,i}=0, i\in[1,n] \Rightarrow x_k
  = \sum_{l=1}^{n} a_{k,l}z_{l} = 0$, so $y_k = x_kv_k = 0 = x_k$; if $v_k = 1
  \Rightarrow x_k = x_kv_k = y_k$. Thus $X = Y$.
\end{proof}

\begin{lemma} \label{lemma-2} $t_{i} = \sum_{j=1}^{n} A_{i,j} \otimes f_{i,j}
  $ in Algorithm~\ref{alg:bfs-dir-compress} (line 5) is equivalent to $t_{i} =
  A_i \otimes f$ in Algorithm~\ref{alg:spmv-bfs} (line 4).
\end{lemma}
\begin{proof}
  In Algorithm~\ref{alg:bfs-dir-compress}, for $\forall j \in [1,n], f_{i,j} =
  f_i \odot V_{j,i}$, according to Lemma~\ref{lemma-1}, $\sum_{j=1}^{n}
  A_{i,j} \otimes f_{i,j} = \sum_{j=1}^{n}A_{i,j} \otimes f_j \odot
  V_{j,i}=\sum_{j=1}^{n} A_{i,j} \otimes f_j = A_i \otimes f = t_i$.
\end{proof}

\section{Algorithm Analysis}
\label{sec:ana}
In this subsection we'd like to analyse the communication and space cost of
the three algorithms in this paper.

\subsection{Communication Cost}
\label{sec:model}
%% To compare the communication cost, we need a communication cost model.  
We study the parallel BFS problem in the message passing model of distributed
computing: every processor has its own local memory, and data exchange between
processors are done by message passing. The time taken to send a message
between any two processors can be modeled as $T(n)=\alpha + n\beta$, where
$\alpha$ is the latency (or startup time) per message, independent of message
size, $\beta$ is the transfer time per byte (inverse of bandwidth), and $n$ is
the number of bytes transfered~\cite{Kumar:2002}. This time cost model is
generally used to model data movement either between levels of a memory
hierarchy or over a network connecting processors. In this paper, we focus on
the latter case. To simplify the analysis, we assume bandwidth cost is much
bigger than latency cost ($n\beta \gg \alpha$), --- as the dataset of
distributed BFS is big, --- therefore $T(n)$ will be dominated by the
bandwidth cost $n\beta$. For a given network, $\beta$ is constant, so the
communication cost is in direct proportion to the message size $n$.
%% We use $O(n)$ to bound the bandwidth cost of a message of $n$ bytes.
Let \textit{communication volume of a processor} $\mathcal{V}_i$ be the size
of all messages communicated on processor $P_i$ in an algorithm. The
\textit{communication volume of an algorithm} is defined as $\mathcal{V}=max
\{ \mathcal{V}_i \mid i \in [1,p] \}$.
%% For $\mathcal{V}=N$, the communication volume of an algorithm is bound to
%% $O(N)$. Since the cost of data movement is dominated by bandwidth cost.
%% Our goal is to minimize the communication volume of parallel BFS algorithm
%% on distributed memory systems, thus minimizing its communication cost. The
%% algorithm design in section~\ref{sec:dir-bfs} will focus on this goal.

The communication volume of MPI collective communication is derived
from~\cite{Thakur:03,Pjesivac-Grbovic:2007}: For $p$ processors, when each
processor needs to broadcast $n/p$ size of message to others, the
communication volume of both allgather and alltoall are $O(n)$.  There are
many algorithms for allgather, for example, ring and recursive
doubling~\cite{Thakur:03}. The time taken for these two algorithm is
$T_{ring}=(p-1)\alpha + {p-1 \over p} n \beta$ and $T_{rec\_dbl} = \log {p}
\alpha + {p-1 \over p} n \beta$, respectively. No matter what algorithm is
used, the bandwidth cost is the same ${p-1 \over p} n \beta$. In
data-intensive applications like BFS, we assume bandwidth cost is much bigger
than latency cost, so its communication volume is bound to $O(n)$. The
communication volume of alltoall can be done in the same
manner~\cite{Pjesivac-Grbovic:2007}.

For graph $G(V,E)$, let $m=|E|$, $n=|V|$, let $d$ be the diameter of the
graph. At each level of BFS, the communication volume of allgather
(Algorithm~\ref{alg:spmv-bfs}, line 7) is $O(n)$; the algorithm will finish at
level $d$. So the communication volume of Algorithm~\ref{alg:spmv-bfs} is $d
\times O(n)$.

For Algorithm~\ref{alg:bfs-compress}, let $C_i(C_i>1)$ be the compression
ratio of the \textit{Compression} function of Algorithm~\ref{alg:bfs-compress}
(line 7) at level $i$, let $C=\frac{1}{d}\sum_{i=1}^{d} \frac{1}{C_i}(C<1)$ be
the compression ratio factor.  The communication volume of
Algorithm~\ref{alg:bfs-compress} is $Cd \times O(n)$.

For Algorithm~\ref{alg:bfs-dir-compress}, let $p$ be the number of the
processors, $e=m/n$ be the average degree of a vertex, and $C'$ be the
compression ratio factor of Algorithm~\ref{alg:bfs-dir-compress}. The
communication volume of Algorithm~\ref{alg:bfs-dir-compress} is $C'd \times
O(n)$. After sieve, a vertex is sent to at most $min(e,p)$ processors in
Algorithm~\ref{alg:bfs-dir-compress} instead of $p$ in
Algorithm~\ref{alg:spmv-bfs} and \ref{alg:bfs-compress}.  Thus
Algorithm~\ref{alg:bfs-dir-compress}'s messages will contain less nonzeros
than Algorithm~\ref{alg:bfs-compress}'s, which leads to a higher compression
ratio and a smaller $C'(C' < C)$.

\subsection{Memory Consumption}

For Algorithm~\ref{alg:spmv-bfs}, the memory consumption of $f$ is $O(n)$;
$t_i$ and $\pi_i$ are $O(n/p)$. So the memory consumption of each processor of
Algorithm~\ref{alg:spmv-bfs} is $O(n)$.

Compared to Algorithm~\ref{alg:spmv-bfs}, Algorithm~\ref{alg:bfs-compress}
replace $f$ with $f'$, the memory consumption of which is at most as
that of $f$, $O(n)$.  So the memory consumption of each processor of
Algorithm~\ref{alg:bfs-compress} is also $O(n)$.

Compared to Algorithm~\ref{alg:bfs-compress},
Algorithm~\ref{alg:bfs-dir-compress} added $V_i$, which costs $O(n)$ memory.
So the memory consumption of Algorithm~\ref{alg:bfs-dir-compress} is also
bound to $O(n)$.

\section{Experimental Results}
\label{sec:exp}

This section presents experimental results for the distributed BFS. 

\subsection{Experiment Setup}

Our performance results is collected on a 512-node multi-core cluster system,
connected by Infiniband of 40 Gb/s. Each node has an SMP architecture with two
Xeon X5650 CPUs (Westmere), which are connected through Intel QuickPath
Interconnect (QPI) of 6.4 GT/s. The Xeon X5650 has six cores, each supports
simultaneous multithreading (SMT) up to two threads. Each node has 24GB
DDR3-1333 RAM. In our experiments we used up to 512 node, or 6,144 cores, to
run the experiment. We use gcc 4.3.4 and MPICH2 1.4.1 to compile our
algorithms. The GNU OpenMP library is used for intra-node threading. See
Table~\ref{table:platform}.

\begin{table}[t]
  \caption{Experiment Platform}
  \label{table:platform}
  \begin{center}
    \begin{tabular}{lc}
      \toprule%
      \textbf{System} &  SMP Cluster\\\midrule
      Number of Nodes & 512 \\
      Number of CPUs / node & 2 \\\bottomrule
      \textbf{Processor} & Intel X5650 \\\midrule
      Number of cores & 6 \\
      Number of threads & 12 \\
      Core frequency & 2.66 GHz \\
      L1 cache size &  384 KB \\
      L2 cache size & 1536 KB \\
      L3 cache size & 12 MB \\
      Memory type & DDR3-1333 \\
      QPI Speed & 6.4 GT/s \\\bottomrule
      \textbf{Interconnect} & Infiniband \\\midrule
      Rate & 40 Gb/sec (4X QDR) \\\bottomrule
    \end{tabular}
  \end{center}
\end{table}

Our algorithms are based on Graph 500 benchmark. Input datasets are generated
use synthetic kronecker graphs \cite{kronecker-graph} which follow power law
distributions: heavy tails for the degree distribution; small diameters; and
densification and shrinking diameters over time. That means most of vertices
has a small number of neighboring vertices and the graph is sparse.  The graph
size is determined by two parameters: ``Scale'' and ``Edge factor'', where the
total number of vertices $N$ equals $2^{Scale}$, and the number of edges, $M =
edgefactor * N$. The default edgefactor is set 16. In order to save space, an
adjacent array (or list) representing sparse graph is transformed into
compressed sparse row (CSR) or column (CSC). We focus on the CSR-based BFS
implementation in Graph 500. In order to compare the performance of Graph 500
implementations across a variety of architectures, a new performance metric is
adopted in Graph 500. Let $time$ be the measured execution time for running
BFS. Let $m$ be the number of input edge tuples within the component traversed
by the search, counting any multiple edges and self-loops. The normalized
performance rate \textit{traversed edges per second (TEPS)} is defined as:
$TEPS = m / time$.

\begin{table}[t]
  \caption{Different BFS algorithms tested}
  \label{table:diff-bfs-tested}
  \begin{center}
    \begin{tabularx}{0.48\textwidth}{ll}
      \toprule%
      Name & Algorithm Details\\\midrule
      \textit{BIT} & Baseline BFS with bitmap (Algorithm~\ref{alg:spmv-bfs}) \\
      \textit{WAH} & BFS with WAH compression (Algorithm~\ref{alg:bfs-compress})\\
      \textit{DIR-WAH} & Directory-based BFS with WAH compression (Algorithm~\ref{alg:bfs-dir-compress})\\
      %% \textit{ZLB} & BFS with Zlib (best speed) compression \\
      %% \textit{DIR-ZLB} & Directory-based BFS with Zlib (best speed) compression \\
      \bottomrule
    \end{tabularx}
  \end{center}
\end{table}
Table~\ref{table:diff-bfs-tested} lists different BFS algorithms tested in our
experiment.

\subsection{Experiment Results}
\label{sec:exp-res}

%% \paragraph*{Overall Performance}

%% \paragraph*{Compression Ratio and Speed}

%% \paragraph*{Filter}

\begin{figure}[t]
  \centering
  \includegraphics[width=0.45\textwidth]{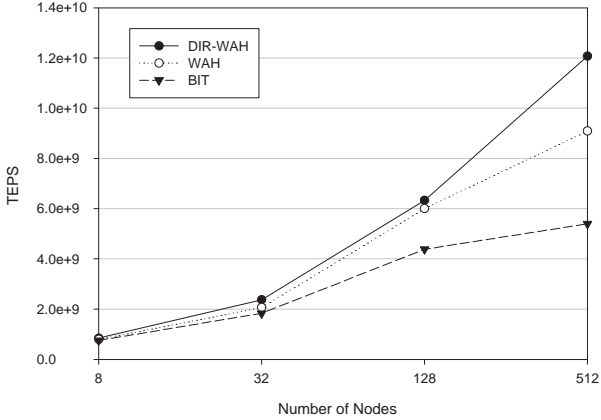}
  \caption{Weak scaling performance of different BFS algorithms. The
    experiment use fixed problem size per node (each node has about 16M
    vertices).}
  \label{fig:weak-scaling}
\end{figure}
Figure~\ref{fig:weak-scaling} shows the weak scaling performance of our BFS
algorithms. We run this experiment on our 512-node SMP cluster, with one
process per SMP node. For intra-node threading, we use the GNU OpenMP library.
%% We choose to use 512 threads, after some runtime tuning, because
%% state-of-the-art processors can hide the memory latency by keeping a number
%% of read requests in flight and improve the performance of data-intensive
%% applications~\cite{Agarwal:2010}.
Algorithm~\ref{alg:bfs-dir-compress} (\textit{DIR-WAH}) outperforms all other
algorithms and have the best scalability.  \textit{DIR-WAH} achieves 1.21E+10
TEPS at scale 33 with 512 nodes, $1.33\times$ than
Algorithm~\ref{alg:bfs-compress} (\textit{WAH}), and $2.24\times$ faster than
Algorithm~\ref{alg:spmv-bfs} (\textit{BIT}). We can see the benefits of
compression and sieve here: with compression, \textit{WAH} is $1.69\times$
faster than \textit{BIT}; with sieve, \textit{DIR-WAH} is another $1.33\times$
than \textit{WAH}. The performance gap between \textit{DIR-WAH} and
\textit{BIT} becomes wider as the number of nodes increases. This is because
the larger the number of nodes is, the more distributed BFS algorithm will
depend communication, and the more benefits compression and sieve will
bring. We will see the time breakdown in the next figure.

\begin{figure}[t]
  \centering
  \includegraphics[width=0.45\textwidth]{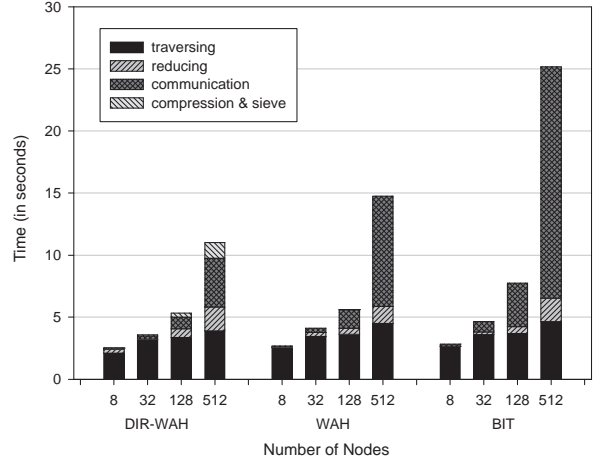}
  \caption{Time breakdown of different BFS algorithms.}
  \label{fig:time-profiling}
\end{figure}
\begin{table}[t]
  \label{table:time-profiling}
  \caption{Description of profiled time in Figure~\ref{fig:time-profiling}.}
  \begin{center}
    \begin{tabular}{ll}
      \toprule%
      Label & Description \\ 
\midrule
      traversing & Local sparse matrix vector multiplication \\
      reducing & MPI reduction to get vertex sum of current frontier\\
      communication & Time spent on communication \\
      compression \& sieve & Time spent on compression and sieve \\
\bottomrule
    \end{tabular}
  \end{center}
\end{table}
Figure~\ref{fig:time-profiling} is the time breakdown of the algorithms in
Figure~\ref{fig:weak-scaling}: ``traversing'' time is the time spent on local
computing; ``reducing'' time is the time spent on a MPI reduction operation to
get the total vertex count of the frontier; ``communicatoin'' time is the time
spent on communication; ``compression \& sieve'' time is the time spent on
compression and sieve.  For all three algorithms, as the number of nodes
increases, ``communication'' times increase exponentially.  For \textit{BIT},
it accounts for as much as $73.2\%$ of the total time for 512 node.  The
``reducing'' times also increases because the imbalance of a graph become more
severe as the graph becomes larger; the local ``traversing'' times remain more
or less the same because the problem size per node is fixed. At 512 node,
\textit{WAH} reduces the ``communication'' time by $52.4\%$ compared to
\textit{BIT}; \textit{DIR-WAH} reduces the ``communication'' time by another
$55.9\%$ compared to \textit{WAH}, achieving a total $79.0\%$ reduction
compared to \textit{BIT}, from 18.6 seconds to 3.9 seconds. On one hand, the
``compression \& sieve'' time of \textit{WAH} (only compression time is
counted for \textit{WAH}) at 512 nodes is less than $0.1\%$ of the total run
time and not shown in the figure. This means the benefit of compression is at
very little cost. On the other hand, the time of ``compression \& sieve'' in
\textit{DIR-WAH}, --- the computing time traded for bandwidth --- accounts for
$11.1\%$ of the total.  This is because Algorithm~\ref{alg:bfs-dir-compress}
(line 9) needs to copy the frontier for each process before sieve. This
copying time is expensive because it is in direct proportion to the number of
processes. Overall, comparing \textit{DIR-WAH} to \textit{WAH} (512 nodes),
sieve costs about 1.3 seconds but saves 5.0 seconds in communication --- the
saving is worth the cost.

\begin{figure}[t]
  \centering
  \includegraphics[width=0.45\textwidth]{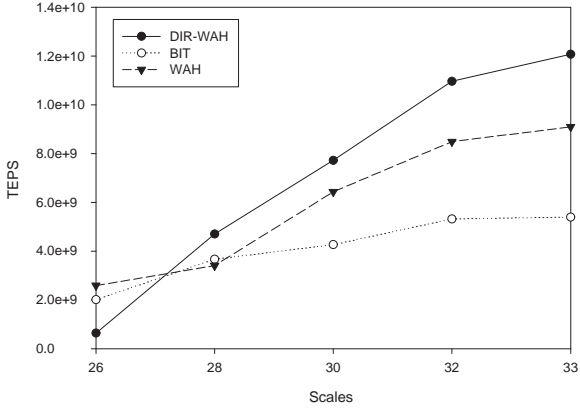}
  \caption{Performance of different BFS algorithms at different scales. The
    experiment runs on 512 nodes.}
  \label{fig:diff-scales}
\end{figure}
Figure~\ref{fig:diff-scales} plots the performance of different BFS algorithms
at different scales. The experiment runs on 512 nodes. We can learn from this
plot that the compression and sieve method favours larger messages. The size
of messages will affect the results: at scale 26, \textit{DIR-WAH},
\textit{WAH} and \textit{BIT} need to exchange 8MB bitmap globally using MPI
collective communications; at scale 33, 1GB. \textit{DIR-WAH} is the slowest
when the scale is small, but it gradually catches up and surpasses all other
algorithms when scale gets bigger. 

\begin{figure}[t]
  \centering
  \includegraphics[width=0.45\textwidth]{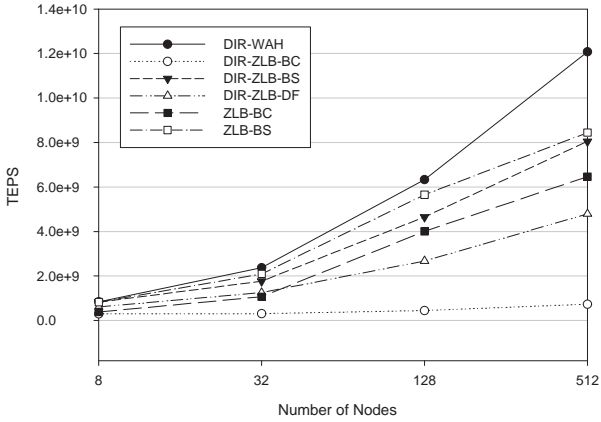}
  \caption{Weak scaling performance result of BFS of different compression
    methods. The experiment use fixed problem size per node (each node has
    about 16M vertices).}
  \label{fig:zlib-weak-scaling}
\end{figure}

\begin{figure}[t]
  \centering
  \includegraphics[width=0.45\textwidth]{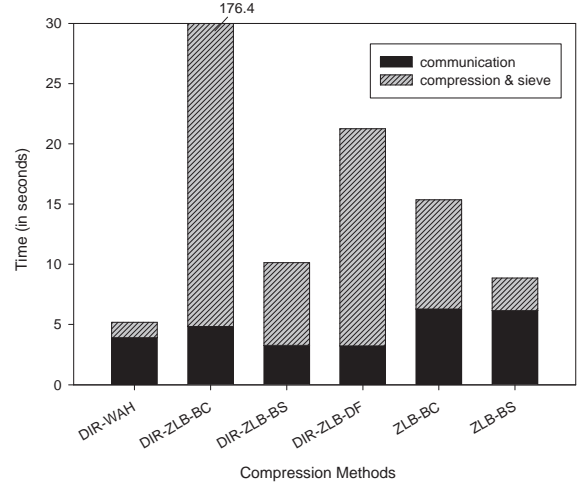}
  \caption{Time profiling of different compression implementations.}
  \label{fig:zlib-time-profiling}
\end{figure}

As mentioned in section~\ref{sec:spmv-bfs}, different methods could be used
for compression. We did not implement all of them but choose two, Zlib
library~\cite{zlib-homepage} and WAH, based on following reasons: Zlib library
is famous for good compression on a wide variety of data and provides
different compression levels; WAH is dedicated to bitmap compression, simpler
than PLWAH and faster than BBC. We use Zlib 1.2.6, and three different
compression levels: best compression (\textit{ZLB-BC}), best speed
(\textit{ZLB-BS}) and default (\textit{ZLB-DF}). The results are plotted in
Figure~\ref{fig:zlib-weak-scaling} and Figure~\ref{fig:zlib-time-profiling}.

Figure~\ref{fig:zlib-weak-scaling} shows the weak scaling performance of BFS
algorithms with different compression and sieve methods. BFS with Zlib best
compression \textit{ZLB-BC} is the slowest. With 512 nodes, \textit{DIR-WAH}
provides the best performance, followed by \textit{ZLB-BS} (69.9\% of
\textit{DIR-WAH}), \textit{DIR-ZLB-BS} (66.7\%), \textit{ZLB-BC} (53.5\%), and
\textit{DIR-ZLB-DF} (39.7\%) respectivelly.

Figure~\ref{fig:zlib-time-profiling} shows the time breakdown of these
algorithms. At scale 33 with 512 nodes, \textit{DIR-ZLB-DF}'s
``communication'' time is the smallest, $0.82\times$ of \textit{DIR-WAH},
followed by \textit{DIR-ZLB-BS} ($0.83\times$), \textit{DIR-ZLB-BC}
($1.23\times$), \textit{ZLB-BS} ($1.57\times$) and \textit{ZLB-BC}
($1.61\times$). Although \textit{DIR-ZLB-DF} and \textit{DIR-ZLB-BS}'s
communication times are less than \textit{DIR-WAH}, their ``compression and
sieve'' times are $14.25\times$ and $5.44\times$ of \textit{DIR-WAH}. So the
overall performance of \textit{DIR-ZLB-DF} and \textit{DIR-ZLB-BS} are worse
than \textit{DIR-WAH}. For all three compression levels in Zlib we tested,
default method, not the best compression method, provides the best compression
ratio. In fact, the Zlib best compression method is not suited for bitmap
compression: it is not only the slowest, but also provides the worst
compression ratio.

\section{Related Works}
\label{sec:related}

%% Frederickson's $O(nm^{1/2})$~\cite{Frederickson:1985} communication
%% complexity is based on the assumption that: 1) the graph is sparse,
%% and 2) BFS could be synchronized $p=n/m^{1/2}$ levels at a
%% time. However, in real world graphs, the diameter $\mathcal{D}$ of the
%% graphs is short, the number of vertices $n$ are huge, and the graph is
%% very sparse that $m \approx n$. So the ideal $p=n/m^{1/2} \approx
%% n^{1/2} \gg \mathcal{D}$. Which means this synchronize every $p$
%% levels approach will not work for low diameter graphs. In contrast,
%% our adaptive approach is synchronized level by level, so it could be
%% applied to low diameter graphs while has a better communication
%% complexity at the same time.

%% The analysis and optimization in this paper targets to parallel BFS
%% algorithms, which is a classical topic extensively studied in computer
%% science. A lot of literatures have been published. Most of these work
%% focus on theoretical algorithm design based on PRAM or its extended
%% model~\cite{pbfs-quinn}. Generally, there are performance gaps between
%% these theoretical model based algorithms and their practice on real
%% machines. We propose a novel asynchronous hybrid BFS
%% algorithm for distributed memory systems with multi-core
%% processors. The new algorithm takes several important architectural
%% features into account and achieve high performance in practice.

Several different approaches are proposed to reduce the communication in
distributed BFS.  Yoo et al.~\cite{Yoo:2005} run distributed BFS on IBM
BlueGene/L with 32,768 nodes. Its high scalability is achieved through a set
of memory and communication optimizations, including a two-dimensional
partitioning of the graph to reduce communication overhead. Bulu\c{c} and
Madduri~\cite{Buluc:2011} improved Yoo et al.'s work by adding hybrid
MPI/OpenMP programming to optimize computation on state-of-the-art multicore
processors, and managed to run distributed BFS on a 40,000-core machine. 
%% They also indicated that the performance of their distributed-memory
%% parallel BFS is heavily dependent on the inter-processor collective
%% communication routines alltoall and allgather.
The method of two-dimensional partitioning reduces the number of processes
involved in collective communications. Our algorithm reduces the communication
overhead in a different way: minimizing the size of messages with compression
and sieve. Moreover, these two optimizations could be combined together to
further reduce the communication cost in distributed BFS. A preliminary result
is presented in Section~\ref{sec:cls} to demonstrate its potential. Beamer et
al.~\cite{Beamer:EECS-2011-117} use a hybrid top-down and bottom-up approach
that dramatically reduces the number of edges examined. The sample code in
Graph 500~\cite{graph500} use bitmap (bitset array) in communication, reducing
its message size. Cong et al.~\cite{Cong:2010} applying communication
coalescing in PGAS implementation to minimize message overhead.  
%% Other efforts has been made on matrix
%% multiplication~\cite{Ballard:2011,Solomonik:2011}.

Benchmarks, algorithms and runtime systems for graph algorithms have gained
much popularity in both academia and industry.  Earlier works on Cray
XMT/MTA~\cite{Bader:2006:MTA,Mizell:2009} and IBM
Cyclops-64~\cite{tan-cyclops64} prove that both massive threads and
fine-grained data synchronization improve BFS performance.  Bader and Madduri
\cite{Bader:2006:MTA} designed a fine-grained parallel BFS which utilizes the
support for hardware threading and synchronization provided by MTA-2, and
ensures that the graph traversal is load-balanced to run on thousands of
hardware threads. Mizell and Maschhoff \cite{Mizell:2009} discussed an
improvement on Cray XMT. Using massive number of threads to hide latency has
long be employed in these specialized multi-threaded machines. With the recent
progress of multi-core and SMT, this technique can be popularized to more
commodity users. Both core-level parallelism and memory-level parallelism are
exploited by Agarwal et al. \cite{Agarwal:2010} for optimized parallel BFS on
Intel Nehalem EP and EX processors. They achieved performances comparable to
special purpose hardwares like Cray XMT and Cray MTA-2 and first identified
the capability of commodity multi-core systems for parallel BFS algorithms.
Scarpazza et al.~\cite{Scarpazza:Cell} use an asynchronous algorithm to
optimize communication between SPE and SPU for running BFS on STI CELL
processors.  Leiserson and Schardl \cite{Leiserson:2010} use Cilk++ runtime
model to implement parallel BFS. Cong et al.~\cite{Cong:2010} present a fast
PGAS implementation of distributed graph algorithms. Another trend is to use
GPU for parallel BFS, for they provide massively parallel hardware threads,
and are more cost-effective than the specialized hardwares.  Generally, GPUs
are good at regular problems with contiguous memory accesses. The challenge of
designing an effective BFS algorithm on GPU is to solve the imbalance between
threads and to hide the cost of data transfer between CPU and GPU.  There are
several works \cite{Hong:2011,Harish-cuda,Luo:2010} working on this direction.

\section{Conclusion}
\label{sec:cls}

The main purpose of this paper is to reduce the communication cost in
distributed breadth-first search (BFS), which is the bottleneck of the
algorithm. We found two problems in previous distributed BFS algorithms:
first, their message formats are not condensed enough; second, broadcasting
messages causes waste.  We propose to reduce the message size by compressing
and sieving.  By compressing the messages, we reduce the communication time by
$52.4\%$.  By sieving the messages with a distributed directory before
compression, we reduce the communication time by another $55.9\%$, achieving a
total $79.0\%$ reduction in communication time and $2.2 \times$ performance
improvement over the baseline implementation.

%% By proposing a distributed directory to replace alltoall communications, we
%% reduce the size of messages communicated significantly.  The key idea of
%% our algorithm is to replicate the directory vectors on both communication
%% sides, use them as sieves to reduce message size, and update both of them
%% locally on each side. We provide a simple analytical cost model,
%% communication volume, to evaluate the communication cost in distributed
%% algorithms. This model is effective for data-intensive applications like
%% BFS when the bandwidth cost is much bigger than latency cost.

For future works, we would like to combine our optimization of message size
with other methods such as two-dimensional partitioning~\cite{Buluc:2011} and
hybrid top-down and bottom-up algorithm~\cite{Beamer:EECS-2011-117}. The
potential is clear. A preliminary optimization of the distributed BFS
algorithm in combinational BLAS library~\cite{combblas-buluc}, compressing the
sparse vector using Zlib library, reduces the communication time by $41.9\%$
and increases overall performance by $1.11\times$. By using compressed bitmap
and adding sieve, we expect to further improve its performance.

\end{document}